\documentclass[11pt]{article}
\usepackage{amssymb,latexsym,amsmath,amsbsy,amsthm}
\usepackage{cite}
\usepackage{stmaryrd}  
\usepackage{arydshln}  
\usepackage{color}
\newtheorem{theo}{Theorem}

\newtheorem{prop}[theo]{Proposition}
\def\nn{\nonumber}

\def\qdots{\mathinner{\mkern1mu\raise1pt\vbox{\kern7pt\hbox{.}}\mkern2mu \raise4pt\hbox{.}\mkern2mu\raise7pt\hbox{.}\mkern1mu}}

\def\N{{\mathbb N}}

\def\gl{\mathfrak{gl}}

\def\h{\mathfrak{h}}
\def\so{\mathfrak{so}}

\def\osp{\mathfrak{osp}}

\def\t{\theta}
\newcommand{\basis}[3]{\left| \!\!\begin{array}{ll}{#1}{#2}\\[-1mm]{#3}\end{array}\!\!\! \right)}
\newcommand{\myatop}[2]{\genfrac{}{}{0pt}{}{#1}{#2}}

\allowdisplaybreaks 
\begin{document}
\begin{center}
{\Large \bf
The parity operator for parafermions and parabosons} \\[5mm] 
{\bf N.I.~Stoilova}\footnote{E-mail: stoilova@inrne.bas.bg}\\[1mm] 
Institute for Nuclear Research and Nuclear Energy, Bulgarian Academy of Sciencies,\\ 
Boul.\ Tsarigradsko Chaussee 72, 1784 Sofia, Bulgaria\\[2mm] 
{\bf J.\ Van der Jeugt}\footnote{E-mail: Joris.VanderJeugt@UGent.be}\\[1mm]
Department of Mathematics, Computer Science and Statistics, Ghent University,\\
Krijgslaan 281-S9, B-9000 Gent, Belgium.
\end{center}

\vskip 2 cm

\begin{abstract}
\noindent 
In this paper we reexamine the definition of parafermions and parabosons by means of Green's triple relations,
and extend these relations by including a parity operator $P$ which is also determined by means of triple relations.
As a consequence, we are dealing with new algebraic structures.
It is shown that the algebra underlying a set of $n$ parafermions together with $P$ is the orthogonal Lie algebra $\so(2n+2)$.
The Fock spaces correspond to particular irreducible representations of $\so(2n+2)$, and the action of $P$ in 
these spaces leads to interesting observations.
Next, we show that the algebra underlying a set of $n$ parabosons together with $P$ is the orthosymplectic 
Lie superalgebra $\osp(2|2n)$.
In this case, the Fock spaces correspond to certain irreducible infinite-dimensional representations of $\osp(2|2n)$.
Both for parafermions and parabosons the spectrum of $P$ is
closely related to the so-called order of statistics $p$, introduced by Green.
\end{abstract}

\vskip 10mm


\setcounter{equation}{0}
\section{Introduction}%

For ordinary bosons or fermions, the parity operator $P$ takes the form $(-1)^F$, where $F$ is the number operator
counting the number of particles in a state of the Fock space. 
Each basis state of the boson or fermion Fock space describes a fixed number of particles, so the parity operator 
takes the value $+1$ for a state with an even number of particles and $-1$ for a state with an odd number of particles.

For parabosons and parafermions, the basis states of the Fock spaces are more complicated, 
and there is no straightforward definition of a number operator.
Since there is more than one definition of parafermions (and parabosons), we should mention that 
in this paper we follow the definition of Green~\cite{Green}.
In Green's approach, the creation and annihilation operators for a set of parafermions or parabosons
are required to satisfy certain cubic or triple relations.
These triple relations generalize the common quadratic anti-commutation relations for fermions 
and quadratic commutation relations for bosons.
Green's triple relations lead to interesting algebraic structures:
the algebra generated by $2n$ creation and annihilation operators subject to the parafermion triple relations
is known to be the Lie algebra $\so(2n+1)=B_n$~\cite{KT, Ryan},
and the algebra generated by $2n$ creation and annihilation operators subject to the paraboson triple relations
is known to be the Lie superalgebra $\osp(1|2n)=B(0,n)$~\cite{Ganchev} 
(we follow Kac’s notation for Lie superalgebras~\cite{Kac}).

The Fock spaces for parafermions and parabosons are characterized or labeled by a positive integer~$p$, 
the so-called ``order of statistics''~\cite{Green, GM}.
For parafermions, these Fock spaces, denoted by $W(p)$, are irreducible finite-dimensional unitary representations 
of the Lie algebra $\so(2n+1)$.
For parabosons, the Fock spaces $V(p)$ are irreducible infinite-dimensional unitary lowest weight representations 
of the Lie superalgebra $\osp(1|2n)$.
The structure of these Fock spaces, and in particular the action of creation and annihilation operators on
basis vectors of the Fock spaces, is complicated.
It was determined in general for parafermions in~\cite{parafermion} and for parabosons in~\cite{paraboson}.

In the current paper we investigate the possibility of introducing a parafermion and paraboson parity operator $P$.
Following the approach of Green, this should be done by postulating some algebraic relations for $P$.
Since the parafermion or paraboson creation and annihilation operators are determined by triple relations,
it is logical to inspect the triple relations satisfied by ordinary fermion/boson creation and annihilation 
operators and their parity operator $P$, and propose these as defining relations for the 
parafermion/paraboson parity operator.
It is important to emphasize that this approach is natural in the context of parafermions and parabosons,
and it leads to an operator $P$ with interesting properties.
But the operator $P$ introduced by means of these algebraic relations does no longer satisfy $P^2=1$ (except in
one representation).
Despite this, we call it the parafermion or paraboson parity operator.

For parafermions (treated in Sections~2 and~3), this approach leads to a set of triple relations for $2n+1$ operators:
the $2n$ parafermion creation and annihilation operators $a_i^\pm$ ($i=1,\ldots,n$) and the parafermion parity operator $P$.
The algebra generated by these $2n+1$ elements subject to the triple relations is examined in Section~2,
and appears to be the Lie algebra $\so(2n+2)=D_{n+1}$.
It is quite surprising that $\so(2n+2)$ is generated by such a small number of generators, subject to simple triple relations.
Next, we turn to the study of the $\so(2n+2)$ parafermion Fock spaces.
These spaces are defined in a natural way, following the definition of $\so(2n+1)$ parafermion Fock spaces.
It turns out that for every positive integer~$p$ there are two Fock spaces $W_+(p)$ and $W_-(p)$, 
closely related to each other.
Our main computational result is the determination of the action of $P$ in a Gelfand-Zetlin (GZ) basis of $W_+(p)$ 
and $W_-(p)$. 
Although this action itself is quite simple, the proof is hard due to the complicated actions of $a_i^\pm$.
The outcome itself is attractive and interesting: the eigenvalues of the parafermion parity operator $P$ 
are related to the order of statistics~$p$. 
The spectrum of $P$ in $W_+(p)$ and in $W_-(p)$ is equal to $\{-p,-p+2,\ldots,p\}$. 
So when the order of statistics $p$ is equal to~1, in which case it is known that parafermions collapse to 
ordinary fermions, the parafermion parity operator yields the eigenvalues $\pm 1$ of the ordinary parity operator.
In order to illustrate the general situation, we give an example of the Fock space for $p=2$, including the 
actions of the parafermion creation and annihilation operators and of the parafermion parity operator.

For parabosons (treated in Sections~4 and~5), we follow a similar approach.
In this case, the algebra generated by $2n$ paraboson creation and annihilation operators $a_i^\pm$
and the paraboson parity operator $P$ is analysed in Section~4 and 
turns out to be the orthosymplectic Lie superalgebra $\osp(2|2n)=C(n+1)$.
Also this is a surprisingly simple mathematical result, namely the algebraic description of $\osp(2|2n)$ by only $2n+1$
generators subject to easy triple relations.
The $\osp(2|2n)$ Fock spaces are defined in Section~5, and for every positive integer~$p$ there
are two Fock spaces $V_+(p)$ and $V_-(p)$. Both of these are irreducible infinite-dimensional unitary lowest weight representations of the Lie superalgebra $\osp(2|2n)$.
Also here, the main computational result is the determination of the action of $P$ 
in a Gelfand-Zetlin basis of $V_+(p)$ and $V_-(p)$. 
As for the parafermion case, the computation is hard but the outcome is simple.
The eigenvalues of the paraboson parity operator $P$ are also related to the order of statistics~$p$,
and given by $\{-p,-p+2,\ldots,p\}$.  
When the order of statistics $p$ is equal to~1, in which case parabosons coincide with bosons, 
the paraboson parity operator yields again the eigenvalues $\pm 1$ of the boson parity operator.

\setcounter{equation}{0}
\section{Parafermions as generators of $\so(2n+2)=D_{n+1}$}%

A system of $n$ fermion creation and annihilation operators $a_i^\pm$ ($i=1,\ldots,n$) is determined by the anti-commutation relations
\begin{equation}
 \{ a_i^+,a_j^+\} = \{ a_i^-,a_j^-\} = 0, \quad \{ a_i^-,a_j^+\} = \delta_{ij}, \quad (i,j\in \{1,\ldots,n\}).
\label{aa}
\end{equation}
The fermion Fock space is generated by a vacuum vector $|0\rangle$ satisfying
\begin{equation}
a_i^- |0\rangle =0\qquad (i\in \{1,\ldots,n\}),  \qquad \langle 0 | 0\rangle = 1,
\end{equation}
and the hermiticity conditions $(a_i^\pm)^\dagger = a_i^\mp$.
A set of orthonormal basis vectors of the Fock space is given by
\begin{equation}
|\t\rangle = |\t_1,\ldots,\t_n\rangle= (a_1^+)^{\t_1}\ldots (a_n^+)^{\t_n} |0\rangle,
\quad \t_i\in\{ 0,1\},
\label{theta}
\end{equation}
and the explicit action of the operators $a_i^\pm$ on this basis is well known:
\begin{align}
 a_i^- |\t\rangle &= \t_i (-1)^{\t_1+\cdots+\t_{i-1}}  |\t_1,\ldots,\t_{i-1},\t_i-1,\t_{i+1},\ldots,\t_n\rangle \nonumber\\
 a_i^+ |\t\rangle &= (1-\t_i) (-1)^{\t_1+\cdots+\t_{i-1}} |\t_1,\ldots,\t_{i-1},\t_i+1,\t_{i+1},\ldots,\t_n\rangle.
\label{f-action}
\end{align}
In this Fock space, the number operator is $F=\sum_{i=1}^n a_i^+a_i^-$, with action
\begin{equation}
F |\t\rangle = (\sum_{i=1}^n \t_i)\, |\t\rangle,
\end{equation}
and the so-called fermion parity operator $P$ takes the form $P=(-1)^F$. $P$ has eigenvalue $+1$ on states with an even number of fermions, 
and $-1$ on states with an odd number of fermions.
It is clear that
\begin{equation}
\{P,P\}=2,\qquad \{ P, a_i^\pm\}=0\qquad (i\in \{1,\ldots,n\}).
\label{ha}
\end{equation}

When the set of $2n+1$ operators $a_i^\pm$ ($i=1,\ldots,n$) and $P$ satisfy the quadratic relations~\eqref{aa} and~\eqref{ha}, 
then it is easy to verify that they also satisfy the following cubic or triple relations:
\begin{align}
& [[a_{i}^{\xi}, a_{j}^{\eta}], a_{k}^{\epsilon}]=
|\epsilon -\eta| \delta_{jk} a_{i}^{\xi} - |\epsilon -\xi| \delta_{ik}a_{j}^{\eta}, 
\label{a-rels} \\
& [[P,a_j^\eta],P]=-4 a_j^\eta, \label{ahh}\\
& [[a_i^\xi,P],a^\epsilon_k]= - |\epsilon-\xi| \delta_{ik} P, \label{aah}
\end{align}
where $i,j,k\in \{1,2,\ldots,n\}$ and $\eta, \epsilon, \xi \in\{+,-\}$ (to be interpreted as $+1$ and $-1$
in the algebraic expressions $\epsilon -\xi$ and $\epsilon -\eta$).

Following the ideas of Green~\cite{Green}, let us now consider a {\em new system} of $2n+1$ generators $a_i^\pm$ ($i=1,\ldots,n$) and $P$ 
which are no longer required to satisfy~\eqref{aa} and~\eqref{ha}, 
but are required to satisfy the triple relations~\eqref{a-rels}-\eqref{aah}.
The elements satisfying~\eqref{a-rels} are known as parafermion creation and annihilation operators~\cite{Green}.
Using similar terminology, we could call the element $P$ satisfying~\eqref{ahh}-\eqref{aah} the parafermion parity operator.

It is well known that the Lie algebra generated by the elements $a_i^\pm$ ($i=1,\ldots,n$) subject to the triple relations~\eqref{a-rels} is the simple Lie algebra $\so(2n+1)=B_n$~\cite{KT, Ryan}.
But what is now the algebraic structure determined by the $2n+1$ generators $a_i^\pm$ ($i=1,\ldots,n$) and $P$ subject to the relations~\eqref{a-rels}-\eqref{aah}?
We have the following result:
\begin{prop}
The Lie algebra generated by the $2n+1$ elements $a_i^\pm$ ($i=1,\ldots,n$) and $P$, subject to the triple relations~\eqref{a-rels}-\eqref{aah} is the simple Lie algebra $\so(2n+2)=D_{n+1}$.
\label{prop1}
\end{prop}

In order to prove this, let us first construct a basis of the Lie algebra generated by these $2n+1$ elements.
Since the relations are triple relations, and of course the Jacobi identity holds, we need to consider only linear and quadratic expressions in the generators.
We denote the independent quadratic elements as follows:
\begin{align}
& b_i^\xi = \frac12\xi\, [P,a_i^\xi] \qquad(\xi=\pm,\ i=1,\ldots,n),\label{bi}\\
& F_{ij}=\frac12 [a_i^+,a_j^-] \qquad (i,j=1,\ldots,n),\\
& P_{ij}=\frac12 [a_i^+,a_j^+] \qquad (1\leq i < j \leq n),\\
& N_{ij}=\frac12 [a_i^-,a_j^-] \qquad (1\leq i < j \leq n).
\end{align}
This brings the total number of basis elements ($a_i^\xi$, $P$, $b_i^\xi$, $F_{ij}$, $P_{ij}$, $N_{ij}$ respectively)
of the Lie algebra to 
\begin{equation}
2n+1 + 2n + n^2 + n(n-1)/2 + n(n-1)/2 = 2n^2+3n+1 = \frac12 (2n+2)(2n+1),
\end{equation}
which is indeed the dimension of $\so(2n+2)$.
Using only the triple relations~\eqref{a-rels}-\eqref{aah}, the Lie brackets between all these basis elements can be determined:
\begin{align}
& [P,b^\xi_i]=2\xi a^\xi_i, \quad [a_i^\xi,b_j^\xi]=0, \quad [a_i^\xi,b_j^{-\xi}]=\xi\delta_{ij} P, \\
& [b_i^\xi, b_j^\eta] =-\xi\eta [a_i^\xi,a_j^\eta] \ (\sim \pm 2 F_{ij}\hbox{ or } \pm 2 P_{ij}\hbox{ or } \pm 2 N_{ij}),\\
& [P,F_{ij}]=0, \quad  [P,P_{ij}]=0, \quad  [P,N_{ij}]=0, \\
& [F_{ij}, a^+_k] = \delta_{jk} a^+_i, \quad [F_{ij}, a^-_k] = -\delta_{ik} a^-_j, \\
& [P_{ij}, a^+_k] = 0, \quad [P_{ij}, a^-_k] = \delta_{jk}a^+_i -\delta_{ik} a^+_j, \\
& [N_{ij}, a^+_k] = \delta_{jk} a^-_i-\delta_{ik} a^-_j, \quad [N_{ij}, a^-_k] = 0, \\
& [F_{ij}, b^+_k] = \delta_{jk} b^+_i, \quad [F_{ij}, b^-_k] = -\delta_{ik} b^-_j, \\
& [P_{ij}, b^+_k] = 0, \quad [P_{ij}, b^-_k] = -\delta_{jk}b^+_i +\delta_{ik} b^+_j, \\
& [N_{ij}, b^+_k] = -\delta_{jk} b^-_i+\delta_{ik} b^-_j, \quad [N_{ij}, b^-_k] = 0, \\
& [F_{ij},F_{kl}] = \delta_{jk}F_{il}-\delta_{il} F_{kj},\\
& [F_{ij},P_{kl}] = \delta_{jk}P_{il}+\delta_{jl} P_{ki},\\
& [F_{ij},N_{kl}] = -\delta_{ik}N_{jl}-\delta_{il} N_{kj},\\
& [N_{ij},P_{kl}] = -\delta_{jk}F_{li}+\delta_{ik} F_{lj}+\delta_{jl} F_{ki}-\delta_{il}F_{kj},\\
& [N_{ij},N_{kl}] = [P_{ij}, P_{kl}]=0.
\end{align}
Note that the elements $b_i^\xi$ satisfy
\begin{equation}
[[b_{i}^{\xi}, b_{j}^{\eta}], b_{k}^{\epsilon}]=
|\epsilon -\eta| \delta_{jk} b_{i}^{\xi} - |\epsilon -\xi| \delta_{ik}b_{j}^{\eta},
\end{equation}
so the $b_i^\xi$ are also parafermions.

Let us now turn to the proof of Proposition~\ref{prop1}.
\begin{proof}
In order to show that the algebra under consideration is equal to $\so(2n+2)$, let us work with the following matrix representation of $\so(2n+2)$.
Let $B$ be the symmetric $(2n+2)\times(2n+2)$ matrix:
\begin{equation}
B=\left( \begin{array}{cccc} 0 & I_n & 0 & 0\\ I_n & 0 & 0 &0\\ 0 & 0 & 0 & 1\\ 0 & 0 & 1 & 0 \end{array} \right),
\end{equation}
where $I_n$ is the identity matrix of size $n\times n$.
The Lie algebra $\so(2n+2)$ can be defined as consisting of the matrices $X$ satisfying
\begin{equation}
X^t B + B X=0,
\end{equation}
where $X^t$ is the notation for the transpose of $X$.
Explicitly, these matrices are of the following form
\begin{equation}
X=\left( \begin{array}{cccc} a & b & x & y \\ d & -a^t & x' & y' \\ -{y'}^t & -y^t & f & 0 \\ -{x'}^t & -x^t & 0 & -f \end{array} \right),
\end{equation}
where $a$ is a $n\times n$ matrix, $b$ and $d$ are antisymmetric $n\times n$ matrices, $x,y,x',y'$ are $n\times 1$ matrices, and $f$ is a number.
Using the common unit matrices $e_{ij}$ (consisting of all zeros, except a 1 at position $(i,j)$), a basis of the Cartan subalgebra $\h$ of $\so(2n+2)$ is given by the elements
\begin{equation}
h_i=e_{ii}-e_{n+i,n+i}\quad (1\leq i\leq n), \quad h_{n+1}=e_{2n+1,2n+1}-e_{2n+2,2n+2}.
\end{equation}
Let us denote, as usual, the basis of $\h^*$ dual to $\{h_1, h_2, \ldots, h_{n+1}\}$ by $\{\epsilon_1,\epsilon_2,\ldots,\epsilon_{n+1}\}$.
Then the following elements are the root vectors of $\so(2n+2)$:
\begin{equation}
\begin{array}{lll}
e_{jk}-e_{n+k,n+j} & 1\leq j,k \leq n,\ j\ne k \qquad & \hbox{root: } \epsilon_j-\epsilon_k \\
e_{j,n+k}-e_{k,n+j} & 1\leq j<k \leq n & \hbox{root: } \epsilon_j+\epsilon_k \\
e_{n+j,k}-e_{n+k,j} & 1\leq j<k \leq n & \hbox{root: } -\epsilon_j-\epsilon_k \\
X_j:=e_{j,2n+1}-e_{2n+2,n+j} \qquad & 1\leq j \leq n & \hbox{root: } \epsilon_j-\epsilon_{n+1}\\
Y_j:=e_{j,2n+2}-e_{2n+1,n+j} & 1\leq j \leq n & \hbox{root: } \epsilon_j+\epsilon_{n+1}\\
X_j':=-e_{n+j,2n+1}+e_{2n+2,j} & 1\leq j \leq n & \hbox{root: } -\epsilon_j-\epsilon_{n+1}\\
Y_j':=-e_{n+j,2n+2}+e_{2n+1,j} & 1\leq j \leq n & \hbox{root: } -\epsilon_j+\epsilon_{n+1}
\end{array}
\end{equation}

Now it is a simple matrix computation to show that the following matrices
\begin{equation}
a^+_i=X_i+Y_i,\quad a^-_i=X_i'+Y_i', \quad P=2 h_{n+1}
\end{equation}
satisfy the triple relations~\eqref{a-rels}, \eqref{ahh} and~\eqref{aah}. 
Hence, this proves Proposition~\ref{prop1}. 
\end{proof}

Note that the generators $a_i^\xi$ of $\so(2n+2)$ are {\em not} root vectors of $\so(2n+2)$, but linear combinations of two root vectors.
The extra generator $P$ is just an element of the Cartan subalgebra.

Observe that Proposition~\ref{prop1} is also interesting purely from the mathematical point of view.
It gives a nice presentation of the Lie algebra $\so(2n+2)$ in terms of $2n+1$ generators subject to a set
of triple relations~\eqref{a-rels}-\eqref{aah}. 
It is an interesting alternative of defining $\so(2n+2)$ in terms of the $3n+3$ Chevalley generators subject to the
Chevalley-Serre relations.

\setcounter{equation}{0}
\section{The $\so(2n+2)$ parafermion Fock space}

The parafermion Fock space $W(p)$ (with $p$ a positive integer) 
is the Hilbert space with unique vacuum vector $|0\rangle$, 
defined by means of ($j,k=1,2,\ldots,n$)~\cite{Green, GM}
\begin{align}
& \langle 0|0\rangle=1, \qquad a_j^- |0\rangle = 0, \qquad (a_j^\pm)^\dagger = a_j^\mp,\nn\\
& [a_j^-,a_k^+] |0\rangle = p\,\delta_{jk}\,|0\rangle,
\label{pFock}
\end{align}
and by irreducibility under the action of the algebra generated by
the elements $a_j^+$, $a_j^-$ ($j=1,\ldots,n$), subject to the triple relations~\eqref{a-rels}.
The parameter $p$ is known as the order of  statistics of the parafermion system.
It is well established that the parafermion Fock space $W(p)$ is the unitary irreducible representation of
$\so(2n+1)$ with lowest weight $(-\frac{p}{2}, -\frac{p}{2},\ldots, -\frac{p}{2})$.
The highest weight of $W(p)$ is $(\frac{p}{2}, \frac{p}{2},\ldots, \frac{p}{2})$, with Dynkin label $[0,\ldots,0,p]$.
A Gelfand-Zetlin basis for $W(p)$ was determined in~\cite{parafermion}, and consists of the vectors
\begin{equation}
 |m)\equiv |m)^n\equiv \left|
\begin{array}{lcllll}
 m_{1n} & \cdots & \cdots & m_{n-1,n} & m_{nn}  \\
 m_{1,n-1} & \cdots & \cdots &  m_{n-1,n-1}  &  \\
\vdots & \qdots & & & \\
m_{11} & & & &
\end{array}
\right) 
= \left| \begin{array}{l} [m]^n \\[2mm] |m)^{n-1} \end{array} \right) \;.
\label{mn}
\end{equation}
The top line of this pattern, also denoted by the $n$-tuple $[m]^n$, is any
partition $\lambda$ (consisting of $n$ non-increasing non-negative integers). 
The remaining $n-1$ lines of the pattern will sometimes be denoted by $|m)^{n-1}$.
All $m_{ij}$ in the above GZ-pattern are non-negative integers,
satisfying the {\em betweenness conditions}
\begin{equation}
m_{i,j+1}\geq m_{ij}\geq m_{i+1,j+1}\qquad (1\leq i\leq j\leq n-1).
\label{between}
\end{equation}
Moreover, these vectors $|m)$ should satisfy $m_{1n}\leq p$. In other words, the top line of $|m)$ is a
partition $\lambda$ with largest part not exceeding $p$.
The action of the Cartan algebra elements $h_k$ ($k=1,\ldots,n$) of $\so(2n+1)$ is~\cite[Eq.~(4.12)]{parafermion}:
\begin{equation}
h_{k}|m)=\left(-\frac{p}{2}+\sum_{j=1}^k m_{jk}-\sum_{j=1}^{k-1} m_{j,k-1}\right)|m), 
\quad (1\leq k\leq n). \label{h_k} \\
\end{equation}
The matrix elements of the parafermion generators,
\begin{equation}
(m'| a_i^\xi | m),
\label{a-action}
\end{equation}
are very complicated, and were determined explicitly in~\cite[Eq.~(4.25)-(4.28)]{parafermion}.

Let us now extend this representation and these actions to representations of $\so(2n+2)$.

In general, it is known~\cite{Murnaghan}
 that the $\so(2n+2)$ representation $W_+(p)$ with Dynkin label $[0,\ldots,0,p]$, or highest weight
$(\frac{p}{2}, \ldots,\frac{p}{2}, \frac{p}{2}) = \sum_{i=1}^{n+1} \frac{p}{2}\epsilon_i$, decomposes to the single $\so(2n+1)$
representation with highest weight $(\frac{p}{2}, \ldots, \frac{p}{2}) = \sum_{i=1}^n \frac{p}{2}\epsilon_i$, i.e.\ to $W(p)$.
The lowest weight of $W_+(p)$ is $(-\frac{p}{2}, \ldots,-\frac{p}{2}, -\frac{p}{2})$ when $n$ is odd, 
and is $(-\frac{p}{2}, \ldots,-\frac{p}{2}, \frac{p}{2})$ when $n$ is even.

Moreover, the $\so(2n+2)$ representation $W_-(p)$ with Dynkin label $[0,\ldots,p,0]$, or highest weight
$(\frac{p}{2}, \ldots, \frac{p}{2},-\frac{p}{2}) = \sum_{i=1}^{n} \frac{p}{2}\epsilon_i-\frac{p}{2}\epsilon_{n+1}$, also decomposes to the single $\so(2n+1)$
representation with highest weight $(\frac{p}{2}, \ldots, \frac{p}{2}) = \sum_{i=1}^n \frac{p}{2}\epsilon_i$, i.e.\ to $W(p)$.
In this case,  the lowest weight of $W_-(p)$ is $(-\frac{p}{2}, \ldots,-\frac{p}{2}, \frac{p}{2})$ when $n$ is odd, 
and is $(-\frac{p}{2}, \ldots,-\frac{p}{2}, -\frac{p}{2})$ when $n$ is even.

Since the representation space $W_\pm(p)$ of $\so(2n+2)$ is the same as the representation space $W(p)$ of $\so(2n+1)$,
we can keep the same notation $|m)$ for the basis vectors of $W_+(p)$ and $W_-(p)$, 
with the same action of the generators $a_i^\xi$ on these vectors $|m)$.
In order to have actions for all generators of $\so(2n+2)$, we only need to determine the action of $P=2h_{n+1}$.
We find the following expression:
\begin{align}
\hbox{in } W_+(p):\quad & P|m) = (-1)^n \left(2\sum_{i=1}^n (-1)^i m_{in} + p\right) \ |m), \label{P-action+}\\
\hbox{in } W_-(p):\quad & P|m) = -(-1)^n \left(2\sum_{i=1}^n (-1)^i m_{in} + p\right) \ |m). \label{P-action-}
\end{align}

\begin{prop}
$W_+(p)$ is the irreducible $\so(2n+2)$ representation with highest weight
$(\frac{p}{2}, \ldots,\frac{p}{2}, \frac{p}{2})$ under the actions~\eqref{a-action}-\eqref{P-action+}.
Similarly, $W_-(p)$ is the irreducible $\so(2n+2)$ representation with highest weight
$(\frac{p}{2}, \ldots,\frac{p}{2}, -\frac{p}{2})$ under the actions~\eqref{a-action}, \eqref{P-action-}.
\label{prop2}
\end{prop}

\begin{proof}
We will sketch the proof only for $W_+(p)$.
For this, it is sufficient to verify that the defining relations~\eqref{a-rels}, \eqref{ahh} and~\eqref{aah} are satisfied in $W_+(p)$.

First of all, note that for the vector $|0)$ (i.e.\ all $m_{ij}=0$ in~\eqref{mn}), \eqref{h_k} and \eqref{P-action+} yield
\begin{equation}
h_k|0)= -\frac{p}{2} |0),\quad h_{n+1}|0) = (-1)^n \frac{p}{2} |0),
\end{equation}
so $|0)$ is indeed the lowest weight vector of weight $(-\frac{p}{2}, \ldots,-\frac{p}{2}, (-1)^n\frac{p}{2})$.
For the vector $|p)$, which is our notation for the vector of type~\eqref{mn} with all $m_{ij}=p$, the above actions give
\begin{equation}
h_k|p)= \frac{p}{2} |p),\quad h_{n+1}|p) = \frac{p}{2} |p),
\end{equation}
so $|p)$ is the highest weight vector of weight $(\frac{p}{2}, \ldots,\frac{p}{2}, \frac{p}{2})$.

Next, let us turn to the actual verification of the actions of~\eqref{a-rels}, \eqref{ahh} and~\eqref{aah} in $W_+(p)$.

The verification of the action of~\eqref{a-rels} is straightforward, since we use the same basis vectors $|m)$ for $W_+(p)$ as for $W(p)$ in $\so(2n+1)$, and for $\so(2n+1)$ this was shown in~\cite{parafermion}.

Since the action of $P$ is diagonal, the verification of the action of~\eqref{ahh} is actually quite simple.
Recall that the matrix element $(m'|a_j^\pm|m)$ is nonzero only if
\begin{equation}
[m']^n = [m]^n_{\pm k} \equiv (m_{1n},\ldots, m_{kn}\pm 1,\ldots,m_{nn})
\end{equation}
for some $k$ in $\{1,2,\ldots,n\}$.
For $[m']^n=[m]^n_{\pm k}$, the matrix element of the left hand side of~\eqref{ahh} is equal to
\begin{equation}
(m'| [[P,a_j^\pm],P]|m)= (m'| 2P a_j^\pm P -a_j^\pm P^2-P^2 a_j^\pm |m).
\end{equation}
If we denote, for convenience, the eigenvalue \eqref{P-action+} of $P$ on $|m)$ as $(-1)^n x$, then we get
\begin{align}
(m'| [[P,a_j^\pm],P]|m)& =  (m'| 2 (x\pm 2(-1)^k)x a_j^\pm -x^2 a_j^\pm - (x\pm 2(-1)^k)^2 a_j^\pm |m)\nonumber\\
& = \left( 2 (x\pm 2(-1)^k)x -x^2 - (x\pm 2(-1)^k)^2 \right) (m'| a_j^\pm |m) \nonumber\\
& = -4 (m'| a_j^\pm |m),
\end{align}
corresponding indeed with the right hand side matrix element of~\eqref{ahh}.

The verification of the action of~\eqref{aah} is more involved, and uses the explicit matrix elements of $a_i^\xi$
determined in~\cite[Eq.~(4.25)-(4.28)]{parafermion}.
Using the notation~\eqref{bi} and the action of $a_i^\xi$ and $P$ on $|m)$,
one can deduce the following.
If 
\begin{equation}
a_i^\pm |m) = \sum_{m'} (m'|a_i^\pm|m) \, |m') \quad\hbox{ where }\quad [m']^n=[m]^n_{\pm j} \hbox{ for some }j,
\end{equation}
then
\begin{equation}
b_i^\pm |m) = \sum_{m'} (-1)^{n+j} (m'|a_i^\pm|m) \, |m') \quad\hbox{ where }\quad [m']^n=[m]^n_{\pm j}.
\end{equation}
So the verification of~\eqref{aah} reduces to verifying
\begin{align}
& [b_i^\xi,a_k^\epsilon]\, |m)=0 \quad\hbox{ for } i\ne k,\label{case1}\\
& [b_i^+,a_i^+] \, |m)=0 , \label{case2}\\
& [b_i^+,a_i^-] \, |m)= 2P\, |m),\label{case3}
\end{align}
since by the symmetry of the known matrix elements $(m'|a_i^\pm|m)$ this also implies that 
$[b_i^-,a_i^-] \, |m)=0$ and $[b_i^-,a_i^+] \, |m)= -2P\, |m)$.
Although it is far from trivial, \eqref{case2} follows from comparing 
the coefficients of equal states in the action of $b_i^+ a_i^+ \, |m)$ and $a_i^+b_i^+ \, |m)$.
Eq.~\eqref{case1} follows from similar considerations.
The computational hard equation to verify is~\eqref{case3}.
Also here, one compares the coefficients of equal states $|m')$ in the action of
$b_i^+a_i^- \, |m)$ and $a_i^-b_i^+ \, |m)$. 
When $|m')\ne |m)$, it follows from the symmetry properties of the explicit matrix elements~\eqref{a-action} that the 
contribution in $a_i^-b_i^+ \, |m)$ is the same as in $b_i^+a_i^- \, |m)$.
Then it remains to work out the diagonal terms, i.e.\ 
the coefficient of $|m)$ in $b_i^+a_i^- \, |m)$ and $a_i^-b_i^+ \, |m)$. 
Using the explicit actions of the parafermion operators in~\cite[Eq.~(4.25)-(4.28)]{parafermion},
such actions always lead to a sum of rational expressions, in which each factor is a linear expression in 
the labels $m_{ij}$. 
A careful analysis of these sums of rational expressions shows that they can always be simplified
using the following rational function identity:
\begin{equation}
\sum_{i=1}^N \ \frac{\displaystyle \prod_{j=0}^N (x_i-y_j)}{\displaystyle \prod_{\myatop{\scriptstyle j=0}{\scriptstyle j\ne i}}^N (x_i-x_j)} = \sum_{i=1}^N (x_i-y_i).
\label{lagrange-id}
\end{equation}
This identity can be proven either using the Lagrange interpolation formula~\cite[see the Appendix]{SVdJ2016}  or using the residue theorem~\cite[Eq.~(32)]{CMP94}.

It would take us too far in this paper to identity the rational expressions that are needed explicitly.
But our analysis has shown that, after simplifications using~\eqref{lagrange-id}, the remaining diagonal term
in $[b_i^+,a_i^-] \, |m)$ does indeed coincide with the action $2P\, |m)$. 
\end{proof}

Let us discuss some general properties of the parafermion parity operator $P$ in the 
$\so(2n+2)$ representation $W_+(p)$ (in $W_-(p)$ the properties are similar, since one can simply replace $P$
by $-P$ to go from $W_+(p)$ to $W_-(p)$).

Using \eqref{P-action+}, we find
\begin{equation}
P |0) = (-1)^n p\,|0),
\end{equation}
and
\begin{equation}
P |p) =  p\,|p),
\end{equation}
where $|p)$ is the highest weight state, consisting of the GZ-pattern~\eqref{mn} with all $m_{ij}=p$.
For other vectors $|m)$ of $W_+(p)$, thus with $m_{1n}\leq p$ and satisfying the betweenness conditions~\eqref{between},
it is easy to deduce from~\eqref{P-action+} that
\begin{equation}
P|m) = x\,|m), \quad\hbox{ where } x\in\{p,p-2,p-4,\ldots,-p+2,-p\}.
\end{equation}
Thus there is a close relation between the order of statistics $p$ and the properties of the parafermion parity operator $P$:
$P$ can assume $p+1$ values, ranging from $-p$ to $p$ in steps of~2.

For the simple case $p=1$, the basis vectors $|m)$ have GZ-patterns consisting of 0's and 1's only, with the 
appropriate conditions~\eqref{between}.
In that case, the action of the parafermion operators $a_i^\xi$ coincides with the action~\eqref{f-action} of fermion operators, under the identification of the GZ-basis vectors $|m)$ with the basis vectors~\eqref{theta}.
The correspondence is as follows: when $|m)$ is a basis vector of $W_+(1)$ with $|m)\ne |0)$, then put
\begin{equation}
\theta_k=\sum_{i=1}^k m_{ik} - \sum_{i=1}^{k-1} m_{i,k-1},
\end{equation}
and then
\begin{equation}
|m) = (-1)^{1+\sum_{i=1}^n m_{in}}\, |\theta_1,\theta_2,\ldots,\theta_n\rangle.
\end{equation}
For $|m)= |0)$, we simply have $|0)=|0,0,\ldots,0\rangle = |0\rangle$.
Note also that 
\begin{equation}
P (a_1^+)^{\t_1}\ldots (a_n^+)^{\t_n} |0\rangle = |\theta_1+\cdots+\theta_n|
 (a_1^+)^{\t_1}\ldots (a_n^+)^{\t_n} |0\rangle,
\end{equation}
so for $p=1$ (the fermion case) $P$ is the common fermion parity operator.

Let us also give a more explicit example, namely the action of all generators of $\so(6)$ ($n=2$) for $p=2$.
It is easy to see that the dimension of $W_+(2)$ is 10.
In Table~1, we list the basis vectors in the leftmost column, and the action of the generators (on top of the Table) on these basis vectors.

\begin{table}[ht]
\label{tab1}
\begin{tabular}{l|l|l|l|l|l|}
 & $a_1^+$ & $a_2^+$ & $a_1^-$ & $a_2^-$ & $P$ \\
\hline \phantom{$\displaystyle\int_a^b$}\hspace{-7mm} 
$\basis{0}{0}{0}$ & $\sqrt{2}\basis{1}{0}{1}$ & $\sqrt{2}\basis{1}{0}{0}$ & $0$ & $0$ & $2\basis{0}{0}{0}$ \\[3mm] 
\hline \phantom{$\displaystyle\int_a^b$}\hspace{-7mm} 
$\basis{1}{0}{0}$ & $\basis{2}{0}{1}-\basis{1}{1}{1}$ & $\sqrt{2}\basis{2}{0}{0}$ & $0$ & $\sqrt{2}\basis{0}{0}{0}$ & $0$ \\[3mm]
\hline \phantom{$\displaystyle\int_a^b$}\hspace{-7mm} 
$\basis{1}{0}{1}$ & $\sqrt{2}\basis{2}{0}{2}$ & $\basis{2}{0}{1}+\basis{1}{1}{1}$ & $\sqrt{2}\basis{0}{0}{0}$ & $0$ & $0$ \\[3mm]
\hline \phantom{$\displaystyle\int_a^b$}\hspace{-7mm} 
$\basis{1}{1}{1}$ & $\basis{2}{1}{2}$ & $\basis{2}{1}{1}$ & $-\basis{1}{0}{0}$ & $\basis{1}{0}{1}$ & $2\basis{1}{1}{1}$ \\[3mm]
\hline \phantom{$\displaystyle\int_a^b$}\hspace{-7mm} 
$\basis{2}{0}{0}$ & $-\sqrt{2}\basis{2}{1}{1}$ & $0$ & $0$ & $\sqrt{2}\basis{1}{0}{0}$ & $-2\basis{2}{0}{0}$ \\[3mm]
\hline \phantom{$\displaystyle\int_a^b$}\hspace{-7mm} 
$\basis{2}{0}{1}$ & $-\basis{2}{1}{2}$ & $\basis{2}{1}{1}$ & $\basis{1}{0}{0}$ & $\basis{1}{0}{1}$ & $-2\basis{2}{0}{1}$ \\[3mm]
\hline \phantom{$\displaystyle\int_a^b$}\hspace{-7mm} 
$\basis{2}{0}{2}$ & $0$ & $\sqrt{2}\basis{2}{1}{2}$ & $\sqrt{2}\basis{1}{0}{1}$ & $0$ & $-2\basis{2}{0}{2}$ \\[3mm]
\hline \phantom{$\displaystyle\int_a^b$}\hspace{-7mm} 
$\basis{2}{1}{1}$ & $-\sqrt{2}\basis{2}{2}{2}$ & $0$ & $-\sqrt{2}\basis{2}{0}{0}$ & $\basis{2}{0}{1}+\basis{1}{1}{1}$ & $0$ \\[3mm]
\hline \phantom{$\displaystyle\int_a^b$}\hspace{-7mm} 
$\basis{2}{1}{2}$ & $0$ & $\sqrt{2}\basis{2}{2}{2}$ & $\basis{1}{1}{1}-\basis{2}{0}{1}$ & $\sqrt{2}\basis{2}{0}{2}$ & $0$ \\[3mm]
\hline \phantom{$\displaystyle\int_a^b$}\hspace{-7mm} 
$\basis{2}{2}{2}$ & $0$ & $0$ & $-\sqrt{2}\basis{2}{1}{1}$ & $\sqrt{2}\basis{2}{1}{2}$ & $2\basis{2}{2}{2}$ \\[3mm]
\hline
\end{tabular}
\caption{Explicit action of $a_i^\pm$ and $P$  (top line) on a basis vector $\basis{m_{12}\,}{m_{22}}{m_{11}}$ (left column) of 
$W_+(2)$.}
\end{table}

Although this is a very easy example, it shows already some particular features where parafermions differ from fermions.
For instance,
\begin{equation}
P a_1^+a_2^+ |0\rangle = \sqrt{2} P \left( \basis{2}{0}{1}-\basis{1}{1}{1} \right)
= -2\sqrt{2} \left( \basis{2}{0}{1}+\basis{1}{1}{1} \right),
\end{equation}
so although $P$ is diagonal in the GZ-basis, it is not diagonal in the parafermion particle basis. 
On the other hand, we do have $P a_1^+ a_2^+ |0\rangle = -2 a_2^+ a_1^+ |0\rangle$ for this example.
The spectrum of $P$ can be read off in the last column of the table.

\setcounter{equation}{0}
\section{Parabosons as generators of $\osp(2|2n)=C(n+1)$}%

In the remaining part of this paper, we will study a set of boson operators, paraboson operators, the paraboson parity operator, and the underlying algebraic structures.

We use the same notation for (para)boson operators as for (para)fermion operators: it is clear from the context which operators we are dealing with.

A system of $n$ boson creation and annihilation operators $a_i^\pm$ ($i=1,\ldots,n$) is determined by the commutation relations
\begin{equation}
 [ a_i^+,a_j^+] = [ a_i^-,a_j^-] = 0, \quad [ a_i^-,a_j^+] = \delta_{ij}, \quad (i,j\in \{1,\ldots,n\}).
\label{1-aa}
\end{equation}
The boson Fock space is generated by a vacuum vector $|0\rangle$ satisfying
\begin{equation}
a_i^- |0\rangle =0\qquad (i\in \{1,\ldots,n\}),  \qquad \langle 0 | 0\rangle = 1,
\end{equation}
and the hermiticity conditions $(a_i^\pm)^\dagger = a_i^\mp$.
A set of orthonormal basis vectors of the Fock space is given by
\begin{equation}
|k\rangle = \frac{1}{\sqrt{k_1!\cdots k_n!}}|k_1,\ldots,k_n\rangle= (a_1^+)^{k_1}\ldots (a_n^+)^{k_n} |0\rangle,
\quad k_i\in\N=\{ 0,1,2\ldots \},
\label{k}
\end{equation}
and the explicit action of the operators $a_i^\pm$ on this basis reads:
\begin{align}
 a_i^- |k\rangle &=  \sqrt{k_i} |k_1,\ldots,k_{i-1},k_i-1,k_{i+1},\ldots,k_n\rangle \nonumber\\
 a_i^+ |k\rangle &=  \sqrt{k_i+1}|k_1,\ldots,k_{i-1},k_i+1,k_{i+1},\ldots,k_n\rangle.
\label{boson-action}
\end{align}
The number operator is $F=\sum_{i=1}^n a_i^+a_i^-$, and its action in the Fock space is
\begin{equation}
F |k\rangle = (\sum_{i=1}^n k_i)\, |k\rangle.
\end{equation}
The boson parity operator $P$ takes the form $P=(-1)^F$,
and has eigenvalue $+1$ on states with an even number of bosons, 
and $-1$ on states with an odd number of bosons.
It satisfies
\begin{equation}
\{P,P\}=2,\qquad \{ P, a_i^\pm \} =0\qquad (i\in \{1,\ldots,n\}).
\label{1-ha}
\end{equation}

As for fermions, we can write down the cubic or triple relations satisfied by $a_i^\pm$ and $P$.
When the set of $2n+1$ operators $a_i^\pm$ ($i=1,\ldots,n$) and $P$ satisfy the quadratic relations~\eqref{1-aa} and~\eqref{1-ha}, it is easy to verify that they also satisfy the following triple relations:
\begin{align}
& [\{a_{i}^{\xi}, a_{j}^{\eta}\}, a_{k}^{\epsilon}]=
(\epsilon -\eta) \delta_{jk} a_{i}^{\xi} + (\epsilon -\xi) \delta_{ik}a_{j}^{\eta}, 
\label{1-a-rels} \\
& [[P,a_j^\eta],P]=-4a_j^\eta, \label{1-ahh}\\
& \{[a_i^\xi,P],a^\epsilon_k\}= - (\epsilon-\xi) \delta_{ik} P, \label{1-aah}
\end{align}
where $i,j,k\in \{1,2,\ldots,n\}$ and $\eta, \epsilon, \xi \in\{+,-\}$.

Following again the approach of Green~\cite{Green}, 
let us consider a {\em new system} of $2n+1$ generators $a_i^\pm$ ($i=1,\ldots,n$) and $P$ 
which are no longer required to satisfy~\eqref{1-aa} and~\eqref{1-ha}, 
but are now required to satisfy the triple relations~\eqref{1-a-rels}-\eqref{1-aah}.
The elements satisfying~\eqref{1-a-rels} are known as paraboson creation and annihilation operators~\cite{Green}.
The element $P$ satisfying~\eqref{1-ahh}-\eqref{1-aah} will be referred to as the paraboson parity operator.

It has been established by Ganchev and Palev~\cite{Ganchev} that the Lie superalgebra 
generated by the odd elements $a_i^\pm$ ($i=1,\ldots,n$) subject to the triple relations~\eqref{1-a-rels} 
is the simple orthosymplectic Lie superalgebra $\osp(1|2n)=B(0,n)$.
Our purpose is to identify the algebraic structure determined by the $2n+1$ generators $a_i^\pm$ ($i=1,\ldots,n$) and $P$ subject to the relations~\eqref{1-a-rels}-\eqref{1-aah} (where $P$ is an even element).
We have the following result:
\begin{prop}
The Lie superalgebra generated by the $2n$ odd elements $a_i^\pm$ ($i=1,\ldots,n$) and the even element $P$, subject to the triple relations~\eqref{1-a-rels}-\eqref{1-aah} is the simple Lie superalgebra $\osp(2|2n)=C(n+1)$.
\label{prop3}
\end{prop}

For the proof, we follow the same ideas as the proof of Proposition~\ref{prop1}.
We first construct a basis of the Lie superalgebra generated by these $2n+1$ elements.
Since all relations are triple relations, we need to consider only linear and quadratic expressions in the generators.
We denote the independent quadratic elements as follows:
\begin{align}
& b_i^\xi = \frac12\xi\, [P,a_i^\xi] \qquad(\xi=\pm,\ i=1,\ldots,n), \label{1-bi} \\
& F_{ij}=\frac12 \{a_i^+,a_j^-\} \qquad (i,j=1,\ldots,n),\\
& P_{ij}=\frac12 \{a_i^+,a_j^+\} \qquad (1\leq i \leq j \leq n),\\
& N_{ij}=\frac12 \{a_i^-,a_j^-\} \qquad (1\leq i \leq j \leq n).
\end{align}
This brings the total number of basis elements ($a_i^\xi$, $P$, $b_i^\xi$, $F_{ij}$, $P_{ij}$, $N_{ij}$ respectively) of the Lie superalgebra to 
\begin{equation}
2n+1 + 2n + n^2 + n(n+1)/2 + n(n+1)/2 = 2n^2+5n+1,
\end{equation}
which is indeed the dimension of $\osp(2|2n)$.
The odd basis elements are $a_i^\pm$ and $b_i^\pm$, the others are even.
Using only the triple relations~\eqref{1-a-rels}, \eqref{1-ahh} and~\eqref{1-aah}, the brackets between all these basis elements can be determined:
\begin{align}
& [P,b^\xi_i]=2\xi a^\xi_i, \quad \{a_i^\xi,b_j^\xi\}=0, \quad \{a_i^\xi,b_j^{-\xi}\}=-\delta_{ij} P, \\
& \{b_i^\xi, b_j^\eta\} =-\xi\eta \{a_i^\xi,a_j^\eta\},\ (\sim 2 F_{ij}\hbox{ or } -2 P_{ij}\hbox{ or } -2 N_{ij}),\\
& [P,F_{ij}]=0, \quad  [P,P_{ij}]=0, \quad  [P,N_{ij}]=0, \label{PFij}\\
& [F_{ij}, a^+_k] = \delta_{jk} a^+_i, \quad [F_{ij}, a^-_k] = -\delta_{ik} a^-_j, \\
& [P_{ij}, a^+_k] = 0, \quad [P_{ij}, a^-_k] = -\delta_{jk}a^+_i -\delta_{ik} a^+_j, \\
& [N_{ij}, a^+_k] = \delta_{jk} a^-_i+\delta_{ik} a^-_j, \quad [N_{ij}, a^-_k] = 0, \\
& [F_{ij}, b^+_k] = \delta_{jk} b^+_i, \quad [F_{ij}, b^-_k] = -\delta_{ik} b^-_j, \\
& [P_{ij}, b^+_k] = 0, \quad [P_{ij}, b^-_k] = \delta_{jk}b^+_i +\delta_{ik} b^+_j, \\
& [N_{ij}, b^+_k] = -\delta_{jk} b^-_i-\delta_{ik} b^-_j, \quad [N_{ij}, b^-_k] = 0, \\
& [F_{ij},F_{kl}] = \delta_{jk}F_{il}-\delta_{il} F_{kj},\\
& [F_{ij},P_{kl}] = \delta_{jk}P_{il}+\delta_{jl} P_{ki},\\
& [F_{ij},N_{kl}] = -\delta_{ik}N_{jl}-\delta_{il} N_{kj},\\
& [N_{ij},P_{kl}] = \delta_{jk}F_{li}+\delta_{ik} F_{lj}+\delta_{jl} F_{ki}+\delta_{il}F_{kj},\\
& [N_{ij},N_{kl}] = [P_{ij}, P_{kl}]=0.
\end{align}
It is easy to see that the $b_i^\xi$ satisfy
\begin{equation}
[\{b_{i}^{\xi}, b_{j}^{\eta}\}, b_{k}^{\epsilon}]=
(\epsilon -\eta) \delta_{jk} b_{i}^{\xi} + (\epsilon -\xi) \delta_{ik}b_{j}^{\eta},
\end{equation}
so the $b_i^\xi$ are also parabosons.

Now we can verify Proposition~\ref{prop3}.
\begin{proof}
In order to show that the algebra under consideration is equal to $\osp(2|2n)$, let us work with the following matrix representation of $\osp(2|2n)$~\cite{Kac}:
\begin{equation}
X=\left( \begin{array}{cccc} f & 0 & x & x' \\ 0 & -f & y & y' \\ {y'}^t & {x'}^t & a & b \\ -{y}^t & -x^t & c & -a^t \end{array} \right),
\end{equation}
where $a$ is a $n\times n$ matrix, $b$ and $d$ are symmetric $n\times n$ matrices, $x,y,x',y'$ are $1\times n$ matrices, and $f$ is a number.
The even matrices have $x=x'=y=y'=0$, and the odd matrices have $a=b=c=0$ and $f=0$.
Using the common unit matrices $e_{ij}$, a basis of the Cartan subalgebra $\h$ of $\osp(2|2n)$ is given by the elements
\begin{equation}
h_{0}=e_{1,1}-e_{2,2}, \quad h_i=e_{2+i,2+i}-e_{2+n+i,2+n+i}\quad (1\leq i\leq n).
\end{equation}
Let us denote, as usual, the basis of $\h^*$ dual to $\{h_0, h_1, \ldots, h_{n}\}$ by $\{\epsilon_0,\epsilon_1,\ldots,\epsilon_{n}\}$.
Then the following elements are the root vectors of $\osp(2|2n)$:
\begin{equation}
\begin{array}{lll}
e_{2+j,2+k}-e_{2+n+k,2+n+j} & 1\leq j,k \leq n,\ j\ne k \qquad & \hbox{root: } \epsilon_j-\epsilon_k \\
e_{2+j,2+n+k}+e_{2+k,2+n+j} & 1\leq j\leq k \leq n & \hbox{root: } \epsilon_j+\epsilon_k \\
e_{2+n+j,2+k}+e_{2+n+k,2+j} & 1\leq j\leq k \leq n & \hbox{root: } -\epsilon_j-\epsilon_k \\
X_j:=e_{1,2+j}-e_{2+n+j,2} \qquad & 1\leq j \leq n & \hbox{root: } \epsilon_0-\epsilon_j\\
Y_j:=e_{2,2+j}-e_{2+n+j,1} & 1\leq j \leq n & \hbox{root: } -\epsilon_0-\epsilon_j\\
X_j':=e_{1,2+n+j}+e_{2+j,2} & 1\leq j \leq n & \hbox{root: } \epsilon_0+\epsilon_j\\
Y_j':=e_{2,2+n+j}+e_{2+j,1} & 1\leq j \leq n & \hbox{root: } -\epsilon_0+\epsilon_j
\end{array}
\end{equation}

It is once more a simple matrix computation to show that the following matrices
\begin{equation}
a^+_i=-X_i'+Y_i',\quad a^-_i=X_i-Y_i, \quad P=2h_{0}
\label{1-axy}
\end{equation}
satisfy the triple relations~\eqref{1-a-rels}, \eqref{1-ahh} and~\eqref{1-aah}. 
Hence, this proves Proposition~\ref{prop3}. 
\end{proof}

The generators $a_i^\xi$ of $\osp(2|2n)$ are {\em not} root vectors of $\osp(2|2n)$, but linear combinations of two root vectors; the extra generator $P$ is an element of the Cartan subalgebra.

Proposition~\ref{prop3} is also of mathematical interest, giving a presentation of the Lie superalgebra $\osp(2|2n)$ in terms of $2n+1$ generators subject to a set of triple relations~\eqref{1-a-rels}-\eqref{1-aah},
as an alternative to the Chevalley generators.

\setcounter{equation}{0}
\section{The $\osp(2|2n)$ paraboson Fock space}

The paraboson Fock space $V(p)$ is the Hilbert space with unique vacuum vector $|0\rangle$, 
defined by means of ($j,k=1,2,\ldots,n$)~\cite{Green, GM}
\begin{align}
& \langle 0|0\rangle=1, \qquad a_j^- |0\rangle = 0, \qquad (a_j^\pm)^\dagger = a_j^\mp,\nn\\
& \{a_j^-,a_k^+\} |0\rangle = p\,\delta_{jk}\,|0\rangle,
\label{1-pFock}
\end{align}
and by irreducibility under the action of the algebra generated by
the elements $a_j^+$, $a_j^-$ ($j=1,\ldots,n$), subject to the triple relations~\eqref{1-a-rels}.
The parameter $p$ (a positive integer) is known as the order of statistics of the paraboson system.
It is known that the paraboson Fock space $V(p)$ is the unitary irreducible representation of
$\osp(1|2n)$ with lowest weight $(\frac{p}{2}, \frac{p}{2},\ldots, \frac{p}{2})$.
A Gelfand-Zetlin basis for $V(p)$ was determined in~\cite{paraboson}, and consists of the vectors
\begin{equation}
 |m)\equiv |m)^n\equiv \left|
\begin{array}{lcllll}
 m_{1n} & \cdots & \cdots & m_{n-1,n} & m_{nn}  \\
 m_{1,n-1} & \cdots & \cdots &  m_{n-1,n-1}  &  \\
\vdots & \qdots & & & \\
m_{11} & & & &
\end{array}
\right) 
= \left| \begin{array}{l} [m]^n \\[2mm] |m)^{n-1} \end{array} \right) \;.
\label{1-mn}
\end{equation}
The notation of these patterns is the same as~\eqref{mn}, and the non-negative integers in~\eqref{1-mn} also 
satisfy the betweenness conditions~\eqref{between}.
The main difference with the vectors~\eqref{mn} is that in the present case the 
top line of $|m)$ is a partition $\lambda$ with at most $p$ nonzero parts, but no restriction on the parts themselves.
This implies that $V(p)$ is infinite-dimensional (just as the Fock space for bosons is infinite-dimensional).
The action of the Cartan algebra elements $h_k$ ($k=1,\ldots,n$) on these vectors of $\osp(1|2n)$ is~\cite[Eq.~(5.4)]{paraboson}:
\begin{equation}
h_{k}|m)=\left(\frac{p}{2}+\sum_{j=1}^k m_{jk}-\sum_{j=1}^{k-1} m_{j,k-1}\right)|m), 
\quad (1\leq k\leq n). \label{1-h_k} \\
\end{equation}
The matrix elements of the paraboson generators,
\begin{equation}
(m'| a_i^\xi | m),
\label{1a-action}
\end{equation}
are again complicated, and were determined explicitly in~\cite[Eq.~(A.9)-(A.12)]{paraboson}.

Let us investigate the possibility of extending these actions to a Fock representation of $\osp(2|2n)$ 
by keeping the same basis, in analogy of the parafermion case, and the same actions of $a_i^\pm$.
Recall that the GZ-basis vectors originate from the decomposition of $V(p)$ in $\gl(n)$ representations,
according to the branching $\osp(1|2n)\supset \gl(n)$, and where the $\gl(n)$ basis is given by the elements $F_{ij}$~\cite{paraboson}.
Since $P$ commutes with all $F_{ij}$, see~\eqref{PFij}, this implies that $P$ is diagonal in the GZ-basis $|m)$.
Let us therefore write
\begin{equation}
P |0\rangle =\lambda_0 |0\rangle, \quad\hbox{ and }\quad Pa_i^+|0\rangle=\lambda_i\, a^+_i|0\rangle.
\label{lambda}
\end{equation}
From Eq.~\eqref{1-aah} with $\xi=-\epsilon=+$ and $k=i$ acting on $|0\rangle$, one finds
\begin{equation}
a^-_ia^+_iP |0\rangle - a^-_iPa^+_i |0\rangle = 2P |0\rangle \quad\Leftrightarrow\quad
\lambda_i p = \lambda_0 (p-2).
\label{eq1}
\end{equation}
Using Eq.~\eqref{1-ahh} with $\eta=+$ and $j=i$ gives
\begin{equation}
2Pa^+_iP -a^+_iP^2 -P^2 a^+_i = -4 a^+_i.
\label{1-aahP}
\end{equation}
Acting on $|0\rangle$ leads to
\begin{equation}
2\lambda_0\lambda_i-\lambda_0^2-\lambda_i^2 =-4,\quad \hbox{ or } \quad (\lambda_i-\lambda_0)^2=4.
\label{eq2}
\end{equation}
The two equations lead to just two possible solutions for $\lambda_0$ and $\lambda_i$:
either $\lambda_0=p$ and $\lambda_i=p-2$, or else $\lambda_0=-p$ and $\lambda_i=-p+2$.
So there are at most two ways to extend the $\osp(1|2n)$ representations $V(p)$ to a representation of $\osp(2|2n)$ with the same basis as $V(p)$. (That it is actually possible to do this, will follow from later considerations.)

Let us denote the representation with $P|0)=p|0)$ by $V_+(p)$, and the one with $P|0)=-p|0)$ by $V_-(p)$.
Following~\eqref{1-axy} and~\eqref{1-h_k}, $V_+(p)$ is a lowest weight representation of $\osp(2|2n)$ with lowest weight  
$(\frac{p}{2}, \frac{p}{2},\ldots, \frac{p}{2})$, 
and $V_-(p)$ is a lowest weight representation of $\osp(2|2n)$ with lowest weight
$(-\frac{p}{2}, \frac{p}{2},\ldots, \frac{p}{2})$.

We claim that both $V_+(p)$ and $V_-(p)$ are representations of $\osp(2|2n)$, both with the same basis
vectors $|m)$ as for $V(p)$, and with the same action of $a_i^\pm$ as for $V(p)$, 
but with the following action of $P=2h_0$:
\begin{align}
\hbox{in } V_+(p):\quad & P|m) = \left(\sum_{i=1}^n (-1)^{m_{in}} + p-n\right) \ |m), \label{1P-action+}\\
\hbox{in } V_-(p):\quad & P|m) = - \left(\sum_{i=1}^n (-1)^{m_{in}} + p-n\right) \ |m). \label{1P-action-}
\end{align}

\begin{prop}
$V_+(p)$ is the irreducible $\osp(2|2n)$ representation with lowest weight
$(\frac{p}{2}, \ldots,\frac{p}{2}, \frac{p}{2})$ under the actions~\eqref{1a-action} and~\eqref{1P-action+}.
Similarly, $V_-(p)$ is the irreducible $\osp(2|2n)$ representation with lowest weight
$(-\frac{p}{2}, \frac{p}{2}, \ldots, \frac{p}{2})$ under the actions~\eqref{1a-action} and \eqref{1P-action-}.
\label{prop4}
\end{prop}

\begin{proof}
Let us give the proof for $V_+(p)$.
For this, it is sufficient to verify that the defining relations~\eqref{1-a-rels}, \eqref{1-ahh} and~\eqref{1-aah} are satisfied in $V_+(p)$.

The verification of the action of~\eqref{1-a-rels} is straightforward, since we use the same basis vectors $|m)$ for $V_+(p)$ as for $V(p)$ in $\osp(1|2n)$, and for $\osp(1|2n)$ this was shown in~\cite{paraboson}.

Since the action of $P$ is diagonal, the verification of the action of~\eqref{1-ahh} is again simple.
Recall that (similar as for parafermions) the paraboson matrix element $(m'|a_j^\pm|m)$ is nonzero only if
\begin{equation}
[m']^n = [m]^n_{\pm k} \equiv (m_{1n},\ldots, m_{kn}\pm 1,\ldots,m_{nn})
\end{equation}
for some $k$ in $\{1,2,\ldots,n\}$.
For $[m']^n=[m]^n_{\pm k}$, the matrix element of the left hand side of~\eqref{1-ahh} is equal to
\begin{equation}
(m'| [[P,a_j^\pm],P]|m)= (m'| 2P a_j^\pm P -a_j^\pm P^2-P^2 a_j^\pm |m).
\end{equation}
If we denote the eigenvalue \eqref{1P-action+} of $P$ on $|m)$ as $x$, we get
\begin{align}
(m'| [[P,a_j^\pm],P]|m)& =  (m'| 2 (x-2(-1)^{m_{kn}})x a_j^\pm -x^2 a_j^\pm - (x-2 (-1)^{m_{kn}})^2 a_j^\pm |m)\nonumber\\
& = \left( 2 (x- 2(-1)^{m_{kn}})x -x^2 - (x-2 (-1)^{m_{kn}})^2 \right) (m'| a_j^\pm |m) \nonumber\\
& = -4 (m'| a_j^\pm |m),
\end{align}
corresponding with the right hand side matrix element of~\eqref{1-ahh}.

The verification of the action of~\eqref{1-aah} is much more involved, and use the explicit paraboson 
matrix elements of $a_i^\xi$ in $V(p)$, determined in~\cite[Eq.~(A.9)-(A.12)]{paraboson}.

As a first step, using the notation~\eqref{1-bi} and the action of $a_i^\xi$ and $P$ on $|m)$,
one can deduce the following.
If 
\begin{equation}
a_i^\pm |m) = \sum_{m'} (m'|a_i^\pm|m) \, |m') \quad\hbox{ where } \quad [m']^n=[m]^n_{\pm j} \hbox{ for some }j,
\end{equation}
then
\begin{equation}
b_i^\pm |m) = \sum_{m'} (-1)^{m_{jn}+1} (m'|a_i^\pm|m) \, |m') \quad\hbox{ where }\quad [m']^n=[m]^n_{\pm j}.
\end{equation}
The actual verification now follows similar steps as the corresponding proof of Proposition~\ref{prop2},
except that it is split in various parts according to whether the GZ-labels of the top line of $|m)$ are even 
or odd. Essentially, the main identity to use is again~\eqref{lagrange-id}, in various disguised forms.
We omit the technical details.
\end{proof}

Let us discuss some general properties of the paraboson parity operator $P$ in the 
$\osp(2|2n)$ representation $V_+(p)$ (in $V_-(p)$ the properties are similar, since one can simply replace $P$
by $-P$ to go from $V_+(p)$ to $V_-(p)$).

Using \eqref{1P-action+}, we find
\begin{equation}
P |0) =  p\,|0).
\end{equation}
In general, from the action~\eqref{1P-action+} it is easy to see that 
\begin{equation}
P|m) = x\,|m), \quad\hbox{ where } x\in\{p,p-2,p-4,\ldots,-p+2,-p\},
\end{equation}
for all allowed GZ-patterns~\eqref{mn} in $V_+(p)$.
Thus there is again a close relation between the order of statistics $p$ 
and the properties of the paraboson parity operator $P$:
$P$ can assume $p+1$ values, ranging from $-p$ to $p$ in steps of~2.

For the simple case $p=1$, the basis vectors $|m)$ have GZ-patterns consisting of 
rows of the form $[m]^j=[m_{1j},0,\ldots,0]$ for each $j$, because of the betweenness conditions~\eqref{between}
and the fact that the length of the partition $(m_{1n},m_{2n},\ldots,m_{nn})$ should be at most $p=1$.
In that case, the action of the paraboson operators $a_i^\xi$ coincides with the action~\eqref{boson-action} of boson operators, under the identification of the GZ-basis vectors $|m)$ with the basis vectors~\eqref{k}.
The correspondence is as follows: 
\begin{equation}
|k_1,k_2,\ldots,k_n\rangle = 
\left|
\begin{array}{lcllll}
k_1+\cdots+k_{n-1}+k_n & 0 & \cdots & 0 & 0  \\
k_1+\cdots+k_{n-1} & 0 & \cdots &  0  &  \\
\vdots & \qdots & & & \\
k_1+k_2 & 0 & & &\\
k_1 & & & &
\end{array}
\right) .
\end{equation}
Of course, in that case~\eqref{1P-action+} gives as eigenvalue
\begin{align}
P|k_1,\ldots,k_n\rangle &= \left(p-n +(-1)^{k_1+\cdots+k_n}+(-1)^0+\cdots+(-1)^0\right) |k_1,\ldots,k_n\rangle\nonumber\\
&= (-1)^{k_1+\cdots+k_n} |k_1,\ldots,k_n\rangle,
\end{align}
and thus $P$ coincides with the boson parity operator.

\setcounter{equation}{0}
\section{Discussion}%

First of all, we want to emphasize the algebraic results of this paper.
Since the early days of Lie algebra theory, it was established that a simple Lie algebra has a nice description 
as an algebra generated by a set of generators subject to a set of relations.
For a simple Lie algebra of rank $n$ it is common to use the $3n$ Chevalley generators 
and the Chevalley-Serre relations. 
These Serre relations are in general cubic relations in the generators (except for $G_2$, 
where some relations are of degree~4).
Also for a basic classical Lie superalgebra of rank~$n$, there is a description in terms of 
$3n$ Chevalley generators and Chevalley-Serre relations~\cite[Section~2.44]{Frappat}. 
Apart from these classical descriptions, some others were known for certain simple Lie (super)algebras, mainly in the 
context of ``creation and annihilation operators''.
As mentioned earlier in this paper, the Lie algebra $B_n$ can be generated by $2n$ creation and 
annihilation operators (parafermions) subject to a set of simple triple (or cubic) relations.
Similary, the Lie superalgebra $B(0,n)$ can be generated by $2n$ creation and 
annihilation operators (parabosons) subject to a set of simple triple relations.
The same property has been studied for the Lie algebra $A_n$~\cite{Palev1977,Jellal} 
and the Lie superalgebra $A(0,n)$~\cite{PSVdJ}.
Several years ago we extended the description of simple Lie algebras and Lie superalgebras by means of $N$ creation
and $N$ annihilation operators as generators, subject to a set of triple relations~\cite{SVdJ2005LA, SVdJ2005LS}.
However, apart from the known cases ($B_n$, $B(0,n)$, $A_n$ and $A(0,n)$), this number $N$ was always larger than
the rank $n$ of the algebra (usually of the order $2n$).
In those studies~\cite{SVdJ2005LA, SVdJ2005LS}, we assumed that the $2N$ creation and annihilation operators should coincide with root vectors of the algebra under consideration.
Now, as a result of the current analysis, we have a description of the Lie algebra $D_{n+1}$ in terms of 
only $2n+1$ generators ($2n$ creation and annihilation operators and one element of the Cartan subalgebra), subject
to a set of triple relations. The creation and annihilation operators however are not root vectors, but a sum of two root
vectors.
Secondly, we have a description of the Lie superalgebra $C(n+1)$ in terms of 
only $2n+1$ generators (again $2n$ creation and annihilation operators and one element of the Cartan subalgebra), subject
to a different set of triple relations. 

As far as physical properties is concerned, we believe that the results of this paper should be of potential interest.
Indeed, in recent years there have been many studies on applications of parafermions or parabosons in 
various areas, such as dark matter and dark energy~\cite{NKMS, KY}, condensed matter physics~\cite{Safonov, WangHazzard}, 
thermodynamics~\cite{Hama, SVdJ2020} and optics~\cite{RJR}.
Proofs about the theoretical detectability of parabosons and parafermions exchanged under the permutation group have been given~\cite{T1, T2, T3, WangHazzard}.
Also experimental realizations of parafermions and parabosons have been proposed~\cite{Huerta1,Huerta2}. 
So far, the complicated structure of Fock spaces $W(p)$ for parafermions and Fock spaces $V(p)$ for parabosons
restrained researchers of using these in physical models, apart from some simple examples or small values for~$p$.
The operator $P$ introduced in this paper, however, does have simple properties in the Fock spaces, and
does also have a very simple spectrum $\{-p,-p+2,\ldots,p\}$
both in $W(p)$ and $V(p)$.
Moreover, it coincides with the usual parity operator for $p=1$, when parafermions (resp.\ parabosons) coincide 
with fermions (resp.\ bosons).
This was the reason to speak of $P$ as the parafermion or paraboson parity operator.
We hope that the simplicity of $P$ and its spectrum helps to bridging the gap from theory to applications
for paraparticles.

\section*{Acknowledgments}
Both authors were supported by the Bulgarian National Science Fund, grant KP-06-N88/3.

\end{document}